\documentclass[final]{amsart}
\usepackage{amssymb}
\usepackage{amsmath}
\usepackage{mathrsfs}
\usepackage{amsthm}
\usepackage[latin1]{inputenc}
\usepackage{color}
\usepackage{colortbl}
\usepackage{graphicx}
\usepackage{float}
\usepackage{color}
\usepackage[comma, authoryear]{natbib}

\usepackage[color, notref, notcite]{showkeys}
\usepackage{enumerate}
\definecolor{labelkey}{rgb}{0.6,0,1}

\newtheorem{thm}{Theorem}[section]
\newtheorem{prop}[thm]{Proposition}
\newtheorem{lemma}[thm]{Lemma}
\newtheorem{cor}[thm]{Corollary}
\newtheorem{hyp}{Assumption}

\theoremstyle{remark}
\newtheorem{rem}[thm]{Remark}

\theoremstyle{definition}

\newcommand{\dN}{\ensuremath{\mathbb{N}}}

\newcommand{\dW}{\ensuremath{\mathbb{W}}}

\newcommand{\cC}{\ensuremath{\mathcal{C}}}

\newcommand{\cF}{\ensuremath{\mathcal{F}}}

\newcommand{\cH}{\ensuremath{\mathcal{H}}}

\newcommand{\cL}{\ensuremath{\mathcal{L}}}

\newcommand{\cQ}{\ensuremath{\mathcal{Q}}}

\newcommand{\cV}{\ensuremath{\mathcal{V}}}
\newcommand{\cW}{\ensuremath{\mathcal{W}}}

\newcommand{\fL}{\ensuremath{\mathfrak{L}}}
\newcommand{\fM}{\ensuremath{\mathfrak{M}}}

\newcommand{\less}{<}
\newcommand{\more}{>}

\newcommand{\N}{\mathbb{N}}
\newcommand{\Px}{\mathbb{P}}
\newcommand{\Ex}{\mathbb{E}}
\newcommand{\R}{\mathbb{R}}

\begin{document}

\title{A level-1 Limit Order book with time dependent arrival rates}
\date{\today}
\author[J.~A. Chávez-Casillas]{Jonathan A. Chávez-Casillas}
\address{Department of Mathematics and Statistics, University of Calgary, Canada}
\email{jonathan.chavezcasil@ucalgary.ca}

\author[R.~J. Elliott]{Robert J. Elliott}
\address{Haskayne School of Business, University of Calgary, Canada, and
Centre for Applied Financial Studies, University of South Australia, Adelaide, Australia}
\email{elliott@ucalgary.ca}

\author[B. R\'emillard]{Bruno R\'emillard}
\address{GERAD, CRM,  and Department of Decision Sciences, HEC
Montr\'eal, 
Canada 
}
\email[Corresponding author]{bruno.remillard@hec.ca}
\date{January 4, 2017}
 \author[A.~V. Swishchuk]{Anatoliy V. Swishchuk}
\address{Department of Mathematics and Statistics, University of Calgary, Canada}
\email{aswish@ucalgary.ca}

\thanks{This research is supported by the Montreal Institute of Structured Finance and Derivatives, the Natural Sciences and Engineering Research Council of Canada, the Social Sciences and Humanities Research Council of Canada, and the Australian Research Council.}

\maketitle
\begin{abstract}
We propose a simple stochastic model for the dynamics of a limit order book, extending the recent work of Cont and de Larrard (2013), where the price dynamics are endogenous, resulting from market transactions. We also show that the conditional diffusion limit of the price process is the so-called Brownian meander.
\end{abstract}

\section{Introduction}

In the now classical approach of financial engineering, one assumes a given model for the price of assets, e.g., geometric Brownian motion, and then uses the  model to evaluate options or optimized portfolios. In this approach, the notion of bid/ask spread is generally not considered and the value of a portfolio is a linear function of the ``price'' of the assets. However, in practice,  the value of a portfolio is not a linear function of the prices. In addition, also in contrast to the classical approach,  the selling value of a portfolio is smaller than the buying value of the same positions. These values are really determined by the so-called limit order book, giving the list of possible bid/ask prices together with the size (number of shares available) at each price.

This limit order book changes rapidly over time, many orders possibly arriving within a millisecond. Either for testing high frequency trading strategies or deciding on an optimal way to buy or sell a large number of shares,
it is important to try to model the behavior of limit order books. Several authors suggested interesting models for limit order books. For example, in \citet{Smith/Farmer/Gillemot/Krishnamurthy:2003}, the authors assumed
that the tick size $\delta$ (least difference between two bid or ask prices)   is constant; this implies that prices are multiples of the tick size. They also assumed that
the markets orders (bid/ask) arrive independently at rate $\mu$ in chunks of $m$ shares; since these orders reduce the number of shares at the best bid or best ask price, they are usually combined with order cancellations. In their model,
the limit orders (bid/ask) also arrive independently at rate $\lambda$ in chunks of $m$ shares; the associated price is said to be selected ``uniformly'' amongst the possible bid prices or ask prices, whatever it means. Basically, they examined some properties of the resulting limit order book, trying to use techniques used in physics to characterize some macro quantities of their model.

More recently,
\citet{Cont/deLarrard:2013} proposed a similar model and they found the asymptotic behavior of the price. In fact, the behaviour  of the asset price is  a consequence of their model for orders arrivals. Contrary to \citet{Smith/Farmer/Gillemot/Krishnamurthy:2003}, they only consider the level-1 order book, meaning that only the best bid and best ask prices are taken into account. In order to do so, they assumed that the bid/ask spread $\delta$ is constant. As before, markets orders for the best bid/ask prices arrive independently at rate $\mu$, in chunks of $m$ shares, and
limit orders for the best bid/ask prices arrive independently at rate $\lambda$, also in chunks of $m$ shares.
When the size (number of shares) of the best bid price attains 0, the bid price decreases by $\delta$ and so does the ask price; the sizes of the best bid/ask prices are then chosen at random from a distribution $\tilde f$. When the size of the best ask price attains 0, the ask price increases by $\delta$ and so does the bid price; the sizes of the best bid/ask prices are then chosen at random from a distribution $f$. With this simple but tractable model, they were able to determine the asymptotic behavior of the price process, instead of assuming it.

According to some participants in the high frequency trading world, the hypothesis of constant arrivals of orders is not justified. Therefore, one should assumed that the arrival rates are time-dependent. This is the model proposed here. We extend the \citet{Cont/deLarrard:2013} setting by assuming that the rates for market orders and limit orders depend on time and that they are also different if they are bid or ask orders. As in \citet{Cont/deLarrard:2013}, under some simple assumptions, we are also able to find the limiting behavior of the price process, and we show how to estimate the main parameters of the model. The main ingredients are
the random times at which the price changes, the associated counting  process, and the distribution of the price changes.

More precisely, in Section \ref{sec:model}, we present the construction of the model we consider. Under some simplifying assumptions, we derive in Section \ref{sec:prop} the distribution of the random times at which the price changes. The asymptotic distribution of the price process is examined in Section \ref{sec:asympt}, while the estimation of the parameters is discussed in Section \ref{sec:est}, together with an example of implementation. The proofs of the main results are given in Appendix \ref{app:proofs}. 

\section{Description of the model}\label{sec:model}

We discuss a level-1 Limit Order Book model using as a framework the model proposed in \cite{Cont/deLarrard:2013}. However, the point processes describing the arrivals of Limit orders have time-dependent periodic rates proportional to the rate describing the arrival of Market orders plus Cancellations.


Recalling the Cont-de Larrard model we will define the level-1  Limit Order book model as follows:
\begin{itemize}
	\item There is just one level on  each side of the order book, i.e., one knows only the best bid and the best ask prices, together with their sizes (number of available shares at these prices).
	\item The spread is constant and always equals the tick size $\delta$.
	\item Order volume is assumed to be constant (set as one unit).
	\item Limit Orders at the bid  and ask sides of the book arrive independently according to  inhomogeneous Poisson  processes $\fL_t^b$ and $\fL_t^a$, with intensities  $\lambda^b_t$ and $\lambda_t^a$ respectively.
	\item Market Orders plus Cancellations at the bid and ask sides of the book arrive independently according to  inhomogeneous Poisson  processes  $\fM_t^b$ and $\fM_t^a$, with intensities $\mu^b_t$ and $\mu^a_t$ respectively.
	\item The processes $\fL_t^a, \fL_t^b, \fM_t^a$ and $\fM_t^b$ are all independent.
	\item Every time there is a depletion at the ask side of the book,  both the bid and the ask prices increase by one tick, and the size of both queues gets redrawn from some distribution $f\in\N^2$.

	\item Every time there is a depletion at the bid side of the book, both the bid and the ask prices decrease by one tick, and the size of both queues gets redrawn from some distribution $\tilde f\in\N^2$.

\end{itemize}

\subsection{Construction of the processes}\label{ssec:model}

First, consider the following infinitesimal generators of birth and death processes:
\begin{equation}\label{eq:generatora}
\left(L_t^a\right)_{ij} = \left\{\begin{array}{cl}
0, & i=0, j\ge 0,\\
\mu_t^a, & 1 \le i, j= i-1,\\
\lambda_t^a, & 1 \le i, j= i+1,\\
-\left(\mu_t^a+\lambda_t^a\right), & 1 \le i, j= i,\\
0, & \text{otherwise}.
\end{array}
\right.
\end{equation}
\begin{equation}\label{eq:generatorb}
\left(L_t^b\right)_{ij} = \left\{\begin{array}{cl}
0, & i=0, j\ge 0,\\
\mu_t^b, & 1 \le i, j= i-1,\\
\lambda_t^b, & 1 \le i, j= i+1,\\
-\left(\mu_t^b+\lambda_t^b\right), & 1 \le i, j= i,\\
0, & \text{otherwise}.
\end{array}
\right.
\end{equation}

Note that $0$ is an absorbing state for any Markov chain with generators $L^a$ or $L^b$. When a chain reaches the absorbing point $0$, one calls it extinction.

\bigskip
To describe precisely the behavior of the price process $S_t$ and the queues sizes process $q_t = (q_t^b,q_t^b)$, one needs to define the following sequence of random times.
Let $\sigma_{x_0}^{(b,1)}$ and  $\sigma_{y_0}^{(a,1)}$ be the extinction times of independent Markov chains $X^{(b,1)}$ and $X^{(a,1)}$ with generators
$L^{(b,1)}$ and  $L^{(a,1)}$, starting from $x_0$ and $y_0$ respectively, where $L_t^{(a,1)} = L_{t}^a$ and $L_t^{(b,1)} = L_{t}^b$. Further set $\tau_0=0$ and $\tau_1 = \min\left(\sigma_x^{(b,1)},\sigma_y^{(a,1)}\right)$.

Having defined $\tau_1, \ldots, \tau_{n-1}$, set
$V_{n-1} = \sum_{k=0}^{n-1}\tau_k$, and let $\sigma_{x_{n-1}}^{(b,n)}$ and  $\sigma_{y_{n-1}}^{(a,n)}$ be the extinction times of independent Markov chains $X^{(b,n)}$ and $X^{(a,n)}$ with generators
$L^{(b,n)}$ and  $L^{(a,n)}$,  starting respectively from $x_{n-1}$ and $y_{n-1}$, where $L_t^{(a,n)} = L_{V_{n-1}+t}^a$ and $L_t^{(b,n)} = L_{V_{n-1}+t}^b$, $t\ge 0$; then set $\tau_n = \min\left(\sigma_{x_{n-1}}^{(n)},\sigma_{y_{n-1}}^{(n)}\right)$. Here the random variables $(x_k,y_k)$ are $\cF_{\tau_{k}}$-measurable, for any $k\ge 0$. In fact, $(x_0,y_0)$ is chosen at random from distribution $f_0$, while $(x_n,y_n)$ is chosen at random from distribution $f_n$ if $\sigma_{x_{n-1}}^{(a,n)} < \sigma_{y_{n-1}}^{(b,n)}$ and chosen at random from distribution $\tilde f_n$ if $\sigma_{x_{n-1}}^{(a,n)} > \sigma_{y_{n-1}}^{(b,n)}$. Now for $t\in [V_{n-1},V_n)$, $q_t^b = X_{t-V_{n-1}}^{(b,n)}$ and $q_t^a=X_{t-V_{n-1}}^{(a,n)}$  starting respectively from $x_{n-1}$ and $y_{n-1}$ at time $V_{n-1}$. Finally,  the price process $S$, representing either the price or the log-price, is defined the following way: for $t\in [V_{n-1},V_n)$, $S_t =  S_{V_{n-1}}$
and $S_{V_{n-1}} = S_{V_{n-2}}+\delta$ if $\sigma_{x_{n-1}}^{(a,n)} < \sigma_{y_{n-1}}^{(b,n)}$ while $S_{V_{n-1}} = S_{V_{n-2}}-\delta$ if $\sigma_{x_{n-1}}^{(b,n)} < \sigma_{y_{n-1}}^{(a,n)}$.


In \cite{Cont/deLarrard:2013}, the authors assumed that the arrivals were time homogeneous, meaning that $L_t^a \equiv  Q^a$ and $L_t^b \equiv  Q^b$. In fact, most of their results were stated for the case $Q^a = Q^b =Q$, where
\begin{equation}\label{eqn:generatorQa}
Q^a_{ij}=\left\{\begin{array}{ccl} 0&\text{if}& i=0,\ j\geq0,\\
\mu^a&\text{if}& 1\leq i,\ j=i-1,\\
\lambda^a&\text{if}& 1\leq i,\ j=i+1,\\
-(\lambda^a+\mu^a)&\text{if}& 1\leq i,\ j=i,\\ 0&\text{if}& |i-j|\more1.\end{array}\right.
\end{equation}
\begin{equation}\label{eqn:generatorQb}
Q^b_{ij}=\left\{\begin{array}{ccl} 0&\text{if}& i=0,\ j\geq0,\\
\mu^b&\text{if}& 1\leq i,\ j=i-1,\\
\lambda^b&\text{if}& 1\leq i,\ j=i+1,\\
-(\lambda^b+\mu^b)&\text{if}& 1\leq i,\ j=i,\\ 0&\text{if}& |i-j|\more1.\end{array}\right.
\end{equation}
and
\begin{equation}\label{eqn:generatorQ}
Q_{ij}=\left\{\begin{array}{ccl} 0&\text{if}& i=0,\ j\geq0,\\
\mu&\text{if}& 1\leq i,\ j=i-1,\\
\lambda&\text{if}& 1\leq i,\ j=i+1,\\
-(\lambda+\mu)&\text{if}& 1\leq i,\ j=i,\\ 0&\text{if}& |i-j|\more1.\end{array}\right.
\end{equation}


\section{Distributional properties}\label{sec:prop}

Because of the independence between the ask and the bid side of the book before the first price change, to analyze the distribution of $\tau_1$, it is enough to study one side of the orderbook, say the ask. In this case, an explicit formula for $\Px[\sigma^{(a,1)}\more t]$ is given in the next section. 

\subsection{Distribution of the inter-arrival time between price changes}


Let $L_{t}$ be the infinitesimal generator of a non homogeneous birth and death process $X$ given by
\begin{equation}\label{eqn:generatorGen}
(L_{t})_{ij} = \left\{\begin{array}{ccl} 0&\text{if}& i=0,\ j\geq0,\\ \mu_t&\text{if}& 1\leq i,\ j=i-1,\\ \lambda_t&\text{if}& 1\leq i,\ j=i+1,\\-(\lambda_t+\mu_t)&\text{if}& 1\leq i,\ j=i,\\ 0&\text{if}& |i-j|\more1.\end{array}\right.
\end{equation}
Notice that 0 is an absorbing state. Also, let $\sigma_X$ be the first hitting times of 0 for this process, i.e.,
\begin{equation}\label{eqn:Defnsigma}
\sigma_X :=\inf\{t\more0|X_t=0\}.
\end{equation}
Then since $0$ is an absorbing state, one has  $\Px_x [ \sigma_X \le t] = \Px_x[X_t =0]$.

%
%
%

It is hopeless to expect solving the problem for general generators so as a first approach, some assumptions  $L^a$ and $L^b$ will be made.

\begin{hyp}\label{hyp:alpha}
There exists a measurable function $\alpha:\R_+\rightarrow\R_+$ such that $A_t = \int_0^t \alpha_s ds <\infty$ for any $t\ge 0$, with $L_t^a = \alpha_t Q^a$ and $L_t^b = \alpha_t Q^b$.
\end{hyp}

\begin{rem}\label{rem:rel:H:Q}
Under the assumption that  $L_t = \alpha_t Q$,  a process  $X$ with infinitesimal generator $L_t$ can be seen as a time change of a process $Y$  with infinitesimal generator $Q$, viz. $X_t= Y_{A_t}$.
In particular, if $\sigma_X$ and $\sigma_Y$ are respectively the first hitting time of $0$ for $X$ and $Y$, then for any $t\ge 0$,
\begin{equation}\label{eq:changetime}
F_L(t;x) := \Px[\sigma_X \le t|\;X_0=x]=\Px[\sigma_Y\le  A_t|\;Y_0=x] := F_Q(A_t;x).
\end{equation}
This result is essential in what follows since it implies  that the distribution of the time between price changes in the present model is comparable to the distribution of the inter-arrival time between price changes for the model considered by \cite{Cont/deLarrard:2013}.
\end{rem}

The following lemma gives the distribution of the extinction time $\sigma_Y$ of a birth and death process $Y$ with generator $Q$.

\begin{lemma}\label{lemma:Sol:u:Cont} Let $Y$ be a birth and death process with generator $Q$ given by \eqref{eqn:generatorQ}. If $\lambda \le \mu$, then
$1-F_Q(t;x) = \Px_x[\sigma_Y \more t]=u_{\lambda,\mu}(t,x)$, where
\begin{equation}\label{eqn:soln:IVP:C}
u_{\lambda,\mu}(t,x)= x \left(\frac{\mu}{\lambda}\right)^{x/2}\int_t^{\infty}\frac{1}{s}I_x\left(2s\sqrt{\lambda\mu}\right)e^{-s(\lambda+\mu)}ds,
\end{equation}
and where $I_\nu(\cdot)$ is the modified Bessel function of the first kind.

If $\lambda>\mu$, then
\begin{equation}\label{eqn:soln:IVP:CNew}
u_{\lambda,\mu}(t,x)= 1 -\left(\frac{\mu}{\lambda}\right)^{x} +  x \left(\frac{\mu}{\lambda}\right)^{x/2}\int_t^{\infty}\frac{1}{s}I_x\left(2s\sqrt{\lambda\mu}\right)e^{-s(\lambda+\mu)}ds.
\end{equation}
In particular,  $\Px_x[\sigma_Y =+\infty ]= 1 -\left(\frac{\mu}{\lambda}\right)^{x} > 0$.
\end{lemma}

\begin{rem}\label{rem:pfexplosion}
The case $\lambda \le \mu$ is proven in  \cite{Cont/deLarrard:2013}.
For the case $\lambda>\mu$, note that  $\Ex_x\left[e^{-s\sigma_Y}\right] = \left(\frac{\lambda+\mu+s-\sqrt{(\lambda+\mu+s)^2-4\lambda\mu}}{2\lambda}\right)^x $, so
letting $s\downarrow 0$ yields $\Px_x(\sigma_Y<\infty) = \left(\frac{\mu}{\lambda}\right)^x$. It then follows that $\Px_x\left[\sigma_Y>t|\sigma_Y<\infty\right] = u_{\mu,\lambda}(t,x)$. Then
$
\Px_x\left[\sigma_Y>t\right] = 1 -\left(\frac{\mu}{\lambda}\right)^{x}  + \left(\frac{\mu}{\lambda}\right)^{x} u_{\mu,\lambda}(t,x)$.
Hence the result.
\end{rem}

It is important to analyze the tail behavior of the survival distribution for $\sigma_Y$. The following lemma, whose proof is deferred to Appendix \ref{app:proofs}, establishes such behavior.
Recall that $\Gamma(s,x) = \int_x^\infty u^{s-1}e^{-u}du$ is the incomplete gamma function.

\begin{lemma} \label{lemma:tail:sigma}
Let $Y$ be a birth and death process with generator $Q$ given by \eqref{eqn:generatorQ}, and assume that $\lambda\le \mu$. Set
 $\mathcal{C}=(\sqrt{\mu}-\sqrt{\lambda})^2$. Then, for a sufficiently large $T$,
\[
\Px[\sigma_Y \more T\;|\;Y_0=x]\sim\left\{\begin{array}{rrl}
\left(\frac{\mu}{\lambda}\right)^{x/2} \frac{x}{\sqrt{\pi\sqrt{\lambda\mu}}} \left[\frac{e^{-T\cC}}{\sqrt{T}} -
\sqrt{\cC}\Gamma\left(\frac{1}{2},T\mathcal{C}\right)\right]&\text{if}& \lambda\less\mu;\\ \frac{x}{\lambda\sqrt{\pi}} \frac{1}{\sqrt{T}}&\text{if}& \lambda=\mu.
\end{array}\right.
\]
Consequently, as expected, if $\lambda=\mu$, $\Ex_x[\sigma_{Y}]=\infty$, whereas if $\lambda\less\mu$,  $\Ex_x\left[e^{\theta \sigma_Y}\right] <\infty$ for $\theta < \cC$. In particular,
$\Ex\left[\sigma_{Y}^k\right]\less\infty$  for every $k\in\N$.

\end{lemma}

\begin{rem}\label{Rem:Cont:incons}
Note that if $\lambda=\mu$, the results in Lemma \ref{lemma:tail:sigma} agree with the results obtained in Eq. (6) in \cite{Cont/deLarrard:2013}. However, if $\lambda\less \mu$, Eq. (5) in \cite{Cont/deLarrard:2013} says that $
\Px[\sigma_Y\more T\;|\;Y_0=x] \sim \frac{x(\lambda+\mu)}{2\lambda(\mu-\lambda)}\frac{1}{T}$,
which is incorrect, since for a birth and death process with death rate larger than its birth rate , the extinction time $\sigma_Y$ has moments of all orders. An easy way to see this is to use the moment generating function  (mgf) computed in Proposition 1 of \cite{Cont/deLarrard:2013} and observe that if $\lambda\less \mu$, then the mgf is defined on an open interval around 0; see, e.g.,
 \cite[Section 21]{Billingsley:1995}.
\end{rem}

Lemma \ref{lemma:Sol:u:Cont} allows a closed formula to be obtained for the distribution of $\sigma_X$, when the rates are proportional to each other, as in Assumption \ref{hyp:alpha}. Such a formula is described in the following proposition, whose proof is deferred to Appendix \ref{app:proofs}.

\begin{prop}\label{prop:Sol:u:A1} Let $X$ be a birth and death process with generator $L$ satisfying  $L_t = \alpha_t Q$. If $\lambda\le \mu$, then the distribution of $\sigma_X$ is given by
\[
\Px_x[\sigma_X \more T]= \Px_x[\sigma_Y \more A_T]= x\left(\frac{\mu}{\lambda}\right)^{x/2}\int_{A_T}^{\infty}\frac{1}{s}I_x\left(2s\sqrt{\lambda\mu}\right)e^{-s(\lambda+\mu)}ds.
\]
\end{prop}

\begin{cor}\label{cor:Distr:tau} Under Assumption \ref{hyp:alpha}, for $A_t=\int_0^t\alpha_sds$, the distribution of $\tau_1$ is given by
\begin{eqnarray*}
\Px_\cL[\tau_1\more T\;|\;  q_0=(x,y) ] &=& \Px_{L^b}[\sigma_x^{(b,1)}\more T] \Px_{L^a}[\sigma_y^{(a,1)}\more T]\\
&=&  \Px_{Q^b}[\sigma_x^{(b,1)}\more A_T] \Px_{Q^b}[\sigma_y^{(a,1)}\more A_T]\\
&=& \Px_\cQ[\tau_1\more A_T\;|\;  q_0=(x,y)].
\end{eqnarray*}
\end{cor}

\begin{proof}
The result follows from the fact that $\tau_1=\sigma_y^{(a,1)}\wedge\sigma_x^{(b,1)}$, Proposition \ref{prop:Sol:u:A1} and the independence between $\sigma_y^{(a,1)}$ and $\sigma_x^{(b,1)}$.
\end{proof}

Now, we present the asymptotic behavior of the survival distribution function of $\tau_1$ under $\cL$.
It follows directly from Lemma \ref{lemma:tail:sigma} and Corollary \ref{cor:Distr:tau}. 

\begin{lemma}\label{lemma:asympt:tau}
Let $\mathcal{C}_a=(\sqrt{\mu^a}-\sqrt{\lambda^a})^2$, $\mathcal{C}_b=(\sqrt{\mu^b}-\sqrt{\lambda^b})^2$, and set $F_\cL(t:x,y) = \Px_\cL\left[\tau_1 \le t\;\Big|\;q_0^b=x, q_0^a=y\right]$, $t\ge 0$. Assume that $\lambda^a \le \mu^a$ and $\lambda^b \le \mu^b$.
Then, as $T\to\infty$, $ 1-F_\cL(T:x,y)$ is asymptotic to
\begin{eqnarray*}
&& \left(\frac{\mu^b}{\lambda^b}\right)^{x/2}  \left(\frac{\mu^a}{\lambda^a}\right)^{y/2}    \frac{xy}{\pi(\lambda^a\lambda^b\mu^a\mu^b)^{1/4} }
 \left[\frac{\exp(-A_{T}\cC_a)}{\sqrt{A_{T}}} - \sqrt{\cC_a} \Gamma\left(\frac{1}{2},A_{T}\cC_a\right)\right] \\
 &&  \qquad\qquad \times \left[\frac{\exp(-A_{T}\cC_b)}{\sqrt{A_{T}}} - \sqrt{\cC_b} \Gamma\left(\frac{1}{2},A_{T}\cC_b\right)\right].
\end{eqnarray*}
In particular,
if $\lambda^a=\mu^a$ and $\lambda^b=\mu^b$, then
$$
A_T\Px_\cL[(\tau_1 > T|q_0=(x,y)] \stackrel{T\to\infty}{\to } \frac{xy}{\pi\sqrt{\lambda^a \lambda^b}}.
$$
\end{lemma}

\begin{rem}\label{rem:bizarre}
It might happen that either $\lambda^a > \mu^a$ or $\lambda^b > \mu^b$. If both these conditions hold, there is a positive probability that the queues will never deplete, so this case must be excluded. There are basically two cases left. The following result follows directly from the proof of Lemma \ref{lemma:asympt:tau}.

\begin{itemize}
\item[(C1)]
 Suppose that $\lambda^b > \mu^b$ and  $\lambda^a \le \mu^a$. Then, as $T\to\infty$, $ 1-F_\cL(T:x,y)$ is asymptotic to
$$
\left[1-\left(\frac{\mu^b}{\lambda^b}\right)^{x} \right] \left(\frac{\mu^a}{\lambda^a}\right)^{y/2}    \frac{y}{\pi(\lambda^a\mu^a)^{1/4} }
 \left[\frac{\exp(-A_{T}\cC_a)}{\sqrt{A_{T}}} - \sqrt{\cC_a} \Gamma\left(\frac{1}{2},A_{T}\cC_a\right)\right].
$$
\item[(C2)]
 Suppose that $\lambda^a > \mu^a$ and  $\lambda^b \le \mu^b$. Then, as $T\to\infty$, $ 1-F_\cL(T:x,y)$ is asymptotic to
$$
\left[1-\left(\frac{\mu^a}{\lambda^a}\right)^{y} \right] \left(\frac{\mu^b}{\lambda^b}\right)^{x/2}    \frac{x}{\pi(\lambda^b\mu^b)^{1/4} }
 \left[\frac{\exp(-A_{T}\cC_b)}{\sqrt{A_{T}}} - \sqrt{\cC_b} \Gamma\left(\frac{1}{2},A_{T}\cC_b\right)\right].
$$
In particular,
if $\lambda^a >\mu^a$ and $\lambda^b=\mu^b$, then
$$
\sqrt{A_T}\Px_\cL[(\tau_1 > T|q_0=(x,y)] \stackrel{T\to\infty}{\to } \frac{x}{\pi\sqrt{\lambda^b}} \left[1-\left(\frac{\mu^a}{\lambda^a}\right)^{y} \right].
$$
\end{itemize}

\end{rem}

\subsection{Probability of a price increase}
In \citet[Proposition 3]{Cont/deLarrard:2013}, the authors considered an asymmetric order flow as given here by the processes $Y^a$ and $Y^b$ for computing the probability of a price increase. This was not used elsewhere in their paper. They obtained the following result, which we cite without much changes. However there are some typos that are corrected here. The proof of the result is given in \citet{VanLeeuwaarden/Raschel:2013}.

\begin{prop} Suppose that $\lambda^a\le \mu^a$ and $\lambda^b\le \mu^b$.
Given $(q^b,q^a)=(x,y)$, the probability $p^{up}(x,y)$ that the next price change is an increase is
\begin{eqnarray*}
p^{up}(x,y) &=& 1-\frac{1}{\pi}\left(\frac{\mu^a}{\lambda^a}\right)^y  \left( \frac{2 \sqrt{\lambda^a\mu^a}}{\mu^a+\lambda^a}\right) \int_0^\pi H_t^x \sin(yt)\sin(t) \\
&& \quad \times \left\{ \frac{2\lambda^b H_t-G_t}{2\frac{\sqrt{\lambda^a\mu^a}}{\mu^a+\lambda^a}\cos(t)-1}\right\}\left\{\frac{1}{\sqrt{G_t^2 - 4\lambda^b\mu^b}}\right\}dt,
\end{eqnarray*}
where $\Sigma = \mu^a+\mu^b+\lambda^a+\lambda^b$,
$G_t = \Sigma-2 \sqrt{\lambda^a\mu^a}\cos(t)$, and
$
H_t = \frac{G_t-\sqrt{G_t^2 - 4\lambda^b\mu^b}}{2\lambda^b}$.
\end{prop}
Under Assumption \ref{hyp:alpha}, the same result applies for our model since $X^a_t = Y^a_{A_t}$ and $X^b_t=Y^b_{A_t}$.
\begin{rem}
One can also use Lemma \ref{lemma:Sol:u:Cont} and Proposition \ref{prop:Sol:u:A1} to obtain the previous result by integration.
\end{rem}

\section{Long-run dynamics of the price process}\label{sec:asympt}

Let $V_n$ be the time of the $n$-th jump in the price, as defined in Section \ref{ssec:model}.
We are interested in analyzing the asymptotic behavior of the number of price changes up to time $t$, that is, in describing the counting  process
\begin{equation}\label{Eqn:Nt:1}
N_t:=\max\{n\ge 0\;|\;V_n\leq t\}, \quad t\ge 0.
\end{equation}

\subsection{Asymptotic behavior of the counting  process $N$}\label{ssec:counting}

The next proposition, depending on a new assumption, whose proof is  deferred to Appendix \ref{app:proofs}, provides an expression which relates the distribution of the partial sums for the waiting times between price changes for the models with the generators $\cL$ and $\cQ$.

\begin{hyp}\label{hyp:ftilde}
      $\sum_{(x,y) \in \dN^2} \tilde f(x,y)\Px_\cQ [ \tau_1 \le t|q_0^b =x, q_0^a = y] = \sum_{(x,y) \in \dN^2 }   f(x,y)\Px_\cQ [ \tau_1 \le t|q_0^b =x, q_0^a = y] = F_{1,\cQ}(t)$. This is true for example, when (i) $\tilde f(x,y)=f(y,x)$ and $Q^a = Q^b$, or (ii) $\tilde f=f$. Properties (i) and (ii) are used for example in \citet{Cont/deLarrard:2013}.
\end{hyp}

\begin{prop}\label{prop:Distr:Sn}
Recall that $A_t=\int_0^t\alpha_sds$.  Then,  under Assumptions \ref{hyp:alpha}--\ref{hyp:ftilde},
\[
\Px_{\cL}[V_n\leq t |\; q_0^b =x, q_0^a = y ]=\Px_\cQ[V_n \leq A_t |\; q_0^b =x, q_0^a = y].
\]

\end{prop}

\begin{rem}
Under generator $\cQ$, $\tau_1, \tau_2, \ldots, \tau_n$ are independent and $\tau_2, \ldots, \tau_n$ are i.i.d. The starting point $(x,y)$ must be random with the correct distribution in order that $\tau_1$ has the same law as $\tau_2$.
\end{rem}

In order  to deal with the counting process $N$,
 we need another assumption.

\begin{hyp}\label{hyp:Average}
There exists a positive constant $\upsilon$ such that $\frac{A_t}{t} \to \upsilon$ as $t\to\infty$.
\end{hyp}

\begin{rem}
Assumption \ref{hyp:Average} is true for example if $\alpha$ is periodic. Such an assumption makes sense. One can easily imagine that $\alpha$ repeats itself everyday. Of course, it must be validated empirically.
One can also suppose that $\alpha$ is random but independent of the other processes. In this case, $\alpha$ would act as a random environment and if we assume that $\alpha$ is stationary and ergodic, then Assumption \ref{hyp:Average} holds almost surely. However, in this case, all  computations are conditional on the environment.
\end{rem}

In order to obtain the asymptotic behavior of the prices, there are two cases to be taken into account:  $\cC_a+\cC_b>0$  and  $\cC_a+\cC_b=0$.

\subsubsection{Case $\cC_a+\cC_b>0$}\label{sssec:C>0}
First, assume that \begin{equation}\label{eqn:mean1}
\gamma_1 = \sum_{(x,y)\in \dN^2}  xy \left(\frac{\mu^b}{\lambda^b}\right)^{x/2}  \left(\frac{\mu^a}{\lambda^a}\right)^{y/2}  f(x,y) <\infty.
\end{equation}
Now, from \citet[p. 376]{Abramowitz/Stegun:1972}, $ I_n(z)= \frac{1}{\pi}\int_0^\pi e^{z \cos{\theta}}\cos(n\theta)d\theta 	$, so
for any $x\in \dN$, $I_n(z)\le e^{z}$.
In this case, it follows from Lemma \ref{lemma:Sol:u:Cont} and Lemma \ref{lemma:tail:sigma} that
\begin{eqnarray*}
\Ex_\cQ(\tau_1) &=& \sum_{(x,y)\in \dN^2}  xy \left(\frac{\mu^b}{\lambda^b}\right)^{x/2}  \left(\frac{\mu^a}{\lambda^a}\right)^{y/2}  f(x,y)  \int_0^\infty \int_0^\infty  t\wedge s g_{x,b}(t) g_{y,a}(s)  dt ds\\
&\le & \frac{\gamma_1}{\max(\cC_a,\cC_b)} <\infty,
\end{eqnarray*}
where $g_{y,a}(s) = \frac{1}{s}I_y\left(2s\sqrt{\lambda^a\mu^a }\right)e^{-s(\lambda^a+\mu^a)}$ and $g_{x,b}(s) = \frac{1}{s}I_x\left(2s\sqrt{\lambda^b\mu^b }\right)e^{-s(\lambda^b+\mu^b)}$.
 Then, under Assumptions \ref{hyp:alpha}--\ref{hyp:ftilde} and under model $\cQ$, $V_n/n \to  \Ex_\cQ(\tau_1) <\infty$ a.s.  Using Assumption \ref{hyp:Average} and Lemma \ref{lemma:asympt:tau}, one then finds that under model $\cL$, $V_n/n$ converges in probability to $ c_1= \Ex_\cQ(\tau_1)/\upsilon$. Finally, using Propositions \ref{prop:laplace}--\ref{prop:counting}, one finds that under $\cL$, $N_t/t$ converges in probability to $\frac{1}{c_1} = \upsilon/\Ex_\cQ(\tau_1)$.
 In addition, $\frac{N_{\lfloor nt\rfloor} - nt/c_1}{\sqrt{n}}  \rightsquigarrow \frac{1}{c_1^{3/2}}\dW(t)$, where $\dW$ is a Brownian motion. This follows from the convergence of $V_n$, under $\cQ$, to a Brownian motian. It also holds under $\cL$, using Assumption \ref{hyp:Average}.

\subsubsection{Case $\cC_a+\cC_b=0$}\label{sssec:C=0}
Assume that \begin{equation}\label{eqn:mean0}
\gamma_0 = \sum_{(x,y)\in \dN^2}  xy   f(x,y) <\infty.
\end{equation}

Then it follows from Lemma \ref{lemma:asympt:tau} and Proposition \ref{prop:psi} that
$$
T\Px_\cL[\tau_1 > T] \stackrel{T\to\infty}{\to } c_0 = \frac{\gamma_0}{\upsilon\pi\sqrt{\lambda^a \lambda^b}}.
$$
As a result, using Propositions \ref{prop:laplace}--\ref{prop:counting} with $f(n) = n\log{n}$, one finds that under $\cL$, $N_t/(t/\log{t})$ converges in probability to $\frac{1}{c_0} = \frac{\upsilon\pi\sqrt{\lambda^a \lambda^b}}{\gamma_0}$.
In particular, if $a_n = n\log{n}$, then $N_{a_n t}/n$ converges in probability to $\frac{t}{c_0}$.
Also, $\frac{V_{\lfloor nt \rfloor}}{n}-c_0 t \log{n} \rightsquigarrow \frac{1}{\upsilon}\cV_t$,
 where $\cV$ is a stable process of index $1$.
It then follows that $\frac{N_{\lfloor n\log{n}  t\rfloor} - nt/c_0}{n/\log{n}}\rightsquigarrow - \frac{1}{c_0 \upsilon}\cV_t $. Note that
$\cV_1$ is the weak limit of $ \frac{V_n}{n}- c_0\upsilon \log{n}$ under $\cQ$, and  $\cV_1 = \tilde \cV_1 +d_0$, where  $d_0$ is the limit of
$ nb_n-c_0\upsilon \log{n}$, where $b_n = \Ex_{\cQ}\{\sin(\tau_1/n)\}$. Next, it follows from \citet{Feller:1971} that  the characteristic function of $\tilde \cV_1$ is $e^{\psi(\zeta)}$, where
$$
\psi(\zeta) = -|\zeta|c_0\upsilon \left\{\frac{\pi}{2}+i{\rm sgn}(\zeta)\log{|\zeta|}\right\}.
$$

\subsection{Asymptotic behavior of the price process}

Under no other additional hypothesis on $f$ and $\tilde f$ than Assumption \ref{hyp:ftilde}, the sequence $(\xi_i)$ of price changes is an ergodic  Markov chain with transition matrix $\Pi$; the sequence
is also independent from $N_t$.
Note that $
P(\xi_2=\delta|\xi_1 =\delta) =\sum_{(i,j) \in \dN^2} f(i,j) P^{up}(i,j)
$ and
$
P(\xi_2=\delta|\xi_1=-\delta) = \sum_{i,j} \tilde f(i,j) P^{up}(i,j)$,
so
 the associated transition matrix $\Pi$  is given by
 $$
 \Pi = \left[ \begin{array}{cc} P(\xi_2=-\delta|\xi_1 =-\delta) &  P(\xi_2=\delta|\xi_1 =-\delta)\\
 P(\xi_2=-\delta|\xi_1 =\delta) & P(\xi_2=\delta|\xi_1 =\delta)\end{array}\right],
 $$
 with stationary distribution $(\nu,1-\nu)$ satisfying
 $$
 \nu = P(\xi_1=-\delta) = \frac{P(\xi_2=-\delta|\xi_1 =\delta)}{P(\xi_2=-\delta|\xi_1 =\delta)+P(\xi_2=\delta|\xi_1 =-\delta)}.
 $$
If $\lfloor c\rfloor$ stands for the largest integer smaller of equal to $c$,   then the sequence  $\cW_n(t) = \frac{1}{\sqrt{n}}\sum_{i=1}^{\lfloor nt\rfloor} \left\{\xi_i - E(\xi_i)\right\}$ converges in law to $\sigma \cW(t)$, where $\cW$ is a Brownian motion, and
 the variance $\sigma^2$ is given by
 \begin{equation}\label{eq:sigma}
 \sigma^2 = 4\delta^2 \left[ \nu(1-\nu) +  \nu \sum_{k=1}^\infty \left\{ (\Pi^k)_{11} -\nu\right\}- (1-\nu)\sum_{k=1}^\infty \left\{(\Pi^k)_{21}-\nu\right\}\right],
 \end{equation}
 with $(\Pi^k)_{ij}$ being the element $(i,j)$ of  $\Pi^k$.

\begin{rem}\label{rem:ftilde2}
If $\tilde f = f$, then   the variables $\xi_j$, $j\ge 1$, are i.i.d. In fact,
$$
P(\xi_2=\delta|\xi_1 =\delta) =\sum_{(i,j) \in \dN^2} f(i,j) P^{up}(i,j)
$$
and
\begin{eqnarray*}
P(\xi_2=\delta|\xi_1=-\delta) &=& \sum_{i,j} \tilde f(i,j) P^{up}(i,j) = \sum_{i,j}  f(i,j) P^{up}(i,j)  \\
& =& P(\xi_2=\delta|\xi_1=\delta).
\end{eqnarray*}
Note also that the variables $\xi_j$, $j\ge 1$,  are independent from $\tau_1, \ldots, \tau_n$. However, unless $Q^a=Q^b$ and $f$ is symmetric, one cannot conclude that $P(\xi_i=\delta)=1/2$.
\end{rem}

 Finally, the price process $S$ can be expressed as
 $$
 S_t = S_0 + \sum_{i=1}^{N_t}\xi_i, \qquad t\ge 0.
 $$
To state the final results, set $a_n = n\log{n}$ or $n$, according as $\cC_a+\cC_b=0$ or not. Then, using the results of Section \ref{ssec:counting},  $N_{a_n t}/n$ converges in probability to $t/c$, where $c=c_0$ or $c=c_1$  according as $\cC_a+\cC_b=0$ or not.
It is then easy to show that
$ n^{-1/2} \sum_{i=1}^{N_{a_n t}}\{\xi_i - E(\xi_1)\}  \rightsquigarrow \frac{\sigma}{\sqrt{c}}\tilde \cW$, where $\tilde \cW$ is a Brownian motion. In fact, for any $t\ge 0$,  $\tilde \cW_t = \sqrt{c}\; \cW_{t/c}$.
Next,
\begin{equation}\label{eq:Sdecomp}
S_{a_n t} - nt/c \Ex(\xi_1)= \sum_{i=1}^{N_{a_n t}}\{\xi_i -\Ex(\xi_1)\} + \Ex(\xi_1) (N_{a_n t} -nt/c).
\end{equation}
This expression shows that there are really two sources of randomness involved in the asymptotic behavior of $S_{a_n t} - nt/c \Ex(\xi_1)$. As before, one must consider the cases $\cC_a+\cC_b>0$ and $\cC_a+\cC_b=0$.

\subsubsection{$\cC_a+\cC_b>0$} In this case,  setting $ W_n(t) = \left\{S_{n t} - nt/c_1 \Ex(\xi_1)\right\}/\sqrt{n}$, then
 $W_n \rightsquigarrow  \tilde \sigma W$, where $W$ is a Brownian motion and
 \begin{equation}\label{eq:sigmatilde}
 \tilde \sigma = \left[\frac{\sigma^2}{c_1}+ \frac{\{E(\xi_1)\}^2}{c_1^3}\right]^{1/2}.
 \end{equation}
In fact, $\tilde \sigma W_t = \frac{\sigma}{\sqrt{c_1}}\tilde \cW_t + \frac{\Ex(\xi_1)}{c_1^{3/2}}\dW_t$, where $\tilde \cW$ and $\dW$ are the two independent Brownian motions appearing respectively in the asymptotic behaviour of the Markov chain and the counting process.
Note that the volatility $\tilde \sigma$ could be estimated by taking the standard deviation of the price increments every $10$ minutes, as proposed in \citet{Cont/deLarrard:2013}; see also \citet{Swishchuk/Cera/Schmidt/Hofmeister:2016}. More generally, if $\Delta$ is the time in seconds between successive prices and $s_\Delta$ is the corresponding standard deviation of the price increments over interval of size $\Delta$, then $\hat {\tilde \sigma} = s_\Delta/\sqrt{\Delta}$.

\subsubsection{$\cC_a+\cC_b=0$} In this case, if $\Ex(\xi_1)=0$, then  using \eqref{eq:Sdecomp}, one obtains that $S_{n\log{n} t}/\sqrt{n} \rightsquigarrow \frac{1}{\sqrt{c_0}}\cW_t $,  where $\cW$ is the Brownian motion resulting from the convergence of the Markov chain.

However, if $\Ex(\xi_1)\neq 0$, then
$(S_{n t} - nt/c_0 \Ex(\xi_1))/(n/\log{n}) \rightsquigarrow - \frac{\Ex(\xi_1)}{c_0 \upsilon}\cV_t  $, where $\cV$ is the stable process defined in Section \ref{sssec:C=0}.

\begin{rem}
Note that in \citet{Cont/deLarrard:2013}, $E(\xi_1)=0$, so the limiting process is a Brownian motion whether $\cC_a+\cC_b=0$ or $\cC_a+\cC_b>0$.
\end{rem}

\subsection{Conditioned limit of the price process}\label{ssec:meander}
If one thinks about it, what one wants to achieve in rescaling the price process $S$ is to replace a discontinuous process by a more amenable process if possible, over a given time interval.  However, on this time interval, the price is known to be positive, so the limiting distribution should be also be positive.

If the unconditioned limit is a Brownian motion, then  the conditioned limit, i.e., conditioning on the fact that the Brownian motion is positive, is called a Brownian meander \citep{Durrett/Iglehart/Miller:1977, Revuz/Yor:1999}.  If the unconditioned limit is a stable process, then the conditioned limit could be called a stable meander. See, e.g., \citet{Caravenna/Chaumont:2008} for more details.
Note that according to \citet{Durrett/Iglehart/Miller:1977}, a Brownian meander $W_t^+$ over $(0,1)$ has conditional density
$$
P(W_t^+ \in dy |W_s^+=x) = \left\{\phi_{t-s}(y-x)-\phi_{t-s}(y+x)\right\}\left\{\frac{\Phi_{1-t}(y) - 1/2}{\Phi_{1-s}(x) - 1/2}\right\},
$$
$ 0< s<t<1$, $x,y >0$, where $\Phi_t$ is the distribution function of a centered Gaussian variable with variance $t$ and associated density $\phi_t$.
It then follows that the infinitesimal generator $\cH_t$ of $W_t^+$ is given by
$$
\cH_t f(x) = f'(x)\{1+\phi_{1-t}(x)\}+ \frac{f''(x)}{2}, \quad x>0.
$$

\section{Estimation of parameters}\label{sec:est}

In order to have identifiable parameters, one has to answer the following question about $\alpha$: What happens if $\alpha$ is multiplied by a positive factor $h$? Then, the value $v$ in Assumption \ref{hyp:Average} is multiplied by $h$. Thus the  parameters $\lambda^a$, $\lambda^b$, $\mu^a$, and $\mu^b$ are all divided by $h$, since for example, $\lambda_t^a = \lambda^a \alpha_t$.
As a result, $E_{\cQ}(\tau_1) $ is then  multiplied by $h$ and so is $\gamma_0$. It then follows that  $c_0$ and $c_1$ are invariant by any scaling. So, one could normalize $\alpha$ so that $v=1$. This is what we will assume from now on. The estimation of the parameters will then be easier.

Next, one of the assumptions of the model is that the size of the orders are constant, which is not the case in practice. So in view of applications, and depending of the statistics of sizes for level-1 orders, if the chosen size is $100$ say, then
an order of size $324$ would count for $3.24$ orders.

Assume that data are collected over a period of $n$ days. Recall that time $0$ corresponds to the opening of the market at 9:30:00 ET.
Let $\Lambda^b_{it}$ and $\Lambda^a_{it}$ be the number of limit orders for bid and ask respectively up to time $t$ (measured in seconds) for day $i$. Further let $t_d$ be the number of seconds considered in a day. Typically, $t_d = 23400$. Finally, let $M^b_{it}$ and $M^a_{it}$ be the number of market orders and cancellations for bid and ask respectively up to time $t$ (measured in seconds) for day $i$.
For any $i\ge 1$, set $v_i= \left\{A_{it_d}-A_{(i-1)t_d} \right\}/t_d $, and set $\hat v = \bar v = \frac{1}{n}\sum_{i=1}^n v_i$. Then for any $i\ge 1$, one should have approximately
\begin{eqnarray*}
\hat\mu^a v_i      &=& M^a_{it_d} /t_d , \qquad \hat\mu^b v_i      = M^b_{it_d} /t_d,\\
\hat\lambda^av_i   &=& \Lambda^a_{it_d} /t_d, \qquad \hat\lambda^b v_i  = \Lambda^b_{it_d} /t_d.
\end{eqnarray*}
Having assumed that $v=1$, one can set
\begin{eqnarray*}
\hat\mu^a         &=& \frac{1}{nt_d}\sum_{i=1}^n M^a_{it_d} , \qquad \hat\mu^b     = \frac{1}{nt_d}\sum_{i=1}^n M^b_{it_d},\\
\hat\lambda^a     &=& \frac{1}{nt_d}\sum_{i=1}^n \Lambda^a_{it_d},\qquad \hat\lambda^b    = \frac{1}{nt_d}\sum_{i=1}^n \Lambda^b_{it_d}.
\end{eqnarray*}
Finally, note that the transition matrix $\Pi$ can be estimated directly from the data, as is $1/c_1$ from $N_t/t$.

\subsection{Example of implementation}

For this example, we use the Facebook data provided  in \citet{Cartea/Jaimungal/Penalva:2015}, from November 3rd, 2014 to November 7th, 2014.
First, the results for the spread are given in Table \ref{tab:spread}, from which we can see that most of the time, the spread $\delta$ is $.01\$$.

\begin{table}[h!]
\caption{Spread distribution  in cents for Facebook, from November 3rd, 2014 to November 7th, 2014.}\label{tab:spread}
\begin{center}
\begin{tabular}{|c|cccccc|}
\hline
  & \multicolumn{5}{c}{Day} & \\
Spread  & 1 & 2 & 3 & 4 &5 & Ave.\\
 $1$ &   91.6\% & 91.8\% & 89.7\% &  88.4\% &  93.6\% &  91.0\% \\
 $2$ &   7.6 \% & 8.0 \% & 10.1\% &  11.1\% &   5.9\% &  8.5\%  \\
  $>2$ & 0.8 \% & 0.2 \% &  0.2\% &   0.5\% &   0.5\% &  0.5\% \\
\hline
\end{tabular}
\end{center}
\end{table}

The values in Table \ref{tab:estimation} can be extracted from Figures \ref{fig:Mu_over_T}--\ref{fig:Lambda_over_T}. It follows that $  \hat\lambda^a < \hat\mu^a$ and $  \hat\lambda^b < \hat \mu^b$,
So with these data, we are in the case where $\cC_a+\cC_b>0$, meaning that the unconditioned limiting price process is a Brownian motion with volatility satisfying \eqref{eq:sigmatilde}.

\begin{rem}
According to Figure \ref{fig:Lambda_over_Mu}, on November 3rd, the ratio $\Lambda_{1t_d}^a/M_{1t_d}^a$ is bigger than one, while the ratio $\Lambda_{1t_d}^b/M_{1t_d}^b$ is smaller than one, meaning that most of the time, the bid queue will be depleted before
the ask queue, so the price has a negative trend throughout that day. This is well illustrated in Figure \ref{fig:midprice}, where it is seen that the price indeed goes down on that day.
\end{rem}

\begin{table}[h!]
\caption{Values of $M_{it_d}^b/t_d$, $ M_{it_d}^a/t_d $, $\Lambda_{it_d}^b/t_d $, and  $\Lambda_{it_d}^a/t_d$.}\label{tab:estimation}
\begin{center}
\begin{tabular}{|c|cccc|}
\hline
Day  &  $\Lambda_{it_d}^b/t_d $&   $\Lambda_{it_d}^a/t_d$ & $M_{it_d}^b/t_d$ &  $ M_{it_d}^a/t_d $    \\
\hline
1 &  494.1500  &   563.2474  & 570.6227  &   553.9348  \\
2 &  610.9476  &   578.6165  & 628.9185  &   613.8630  \\
3 &  661.5511  &   658.3967  & 719.7569  &   672.8735  \\
4 &  398.4293  &   401.4344  & 404.4485  &   415.3457  \\
5 &  427.9106  &   440.4546  & 447.9598  &   458.7763  \\
\hline
ave. & 518.5977 & 528.4299   & 554.3413 &  542.9587\\
 & $\hat \lambda^b $ & $\hat \lambda^a $ & $\hat \mu^b $ & $\hat \mu^a $  \\
\hline
\end{tabular}
  \end{center}
 \end{table}

There are basically two ways of estimating $\tilde\sigma$. One can use the standard deviation of high-frequency data, as exemplified in Table \ref{tab:estimationsigma}, or we could use the analytic expression, as
proposed in \citet{Swishchuk/Vadori:2017, Swishchuk/Cera/Schmidt/Hofmeister:2016}.
\begin{table}[h!]
\caption{Estimation of $\tilde \sigma = s_\Delta/\sqrt{\Delta}$ using high-frequency standard deviations.}\label{tab:estimationsigma}
\begin{center}
\begin{tabular}{|c|ccc|}
\hline
  & \multicolumn{3}{|c|}{$\Delta$}   \\
Day  &  10-minute&   5-minute & 1-minute     \\
\hline
1    &    0.0040  &  0.0052    &   0.0057       \\
2    &    0.0079  &  0.0073    &   0.0075          \\
3    &    0.0069  &  0.0070    &   0.0082          \\
4    &    0.0071  &  0.0062    &   0.0059          \\
5    &    0.0038  &  0.0040    &   0.0051          \\
\hline
pooled &    0.0062  &  0.0060    &   0.0066\\
\hline
\end{tabular}
  \end{center}
 \end{table}
To estimate $\tilde\sigma$ analytically, one needs the estimation of the transition matrix $\Pi$. With the data set, we get
$\hat \Pi =
\left[
\begin{array}{cc}
    0.4731177 & 0.5268512	\\
    0.5241391  & 0.475891
    \end{array}\right]$.
It then follows that $\hat \nu = 0.4987$, so  $E(\xi_1) = 0.0026$, and using formula \eqref{eq:sigma}, one obtains $\hat\sigma = 0.0066$.
Next,  $1/\hat c_1 = 0.6194786$, so $\hat{\tilde \sigma} =  0.0053$, which is quite close to the pooled values in Table \ref{tab:estimationsigma}.

\newpage

\begin{figure}[h!]
\begin{center}
\includegraphics[height=2in,width=4in]{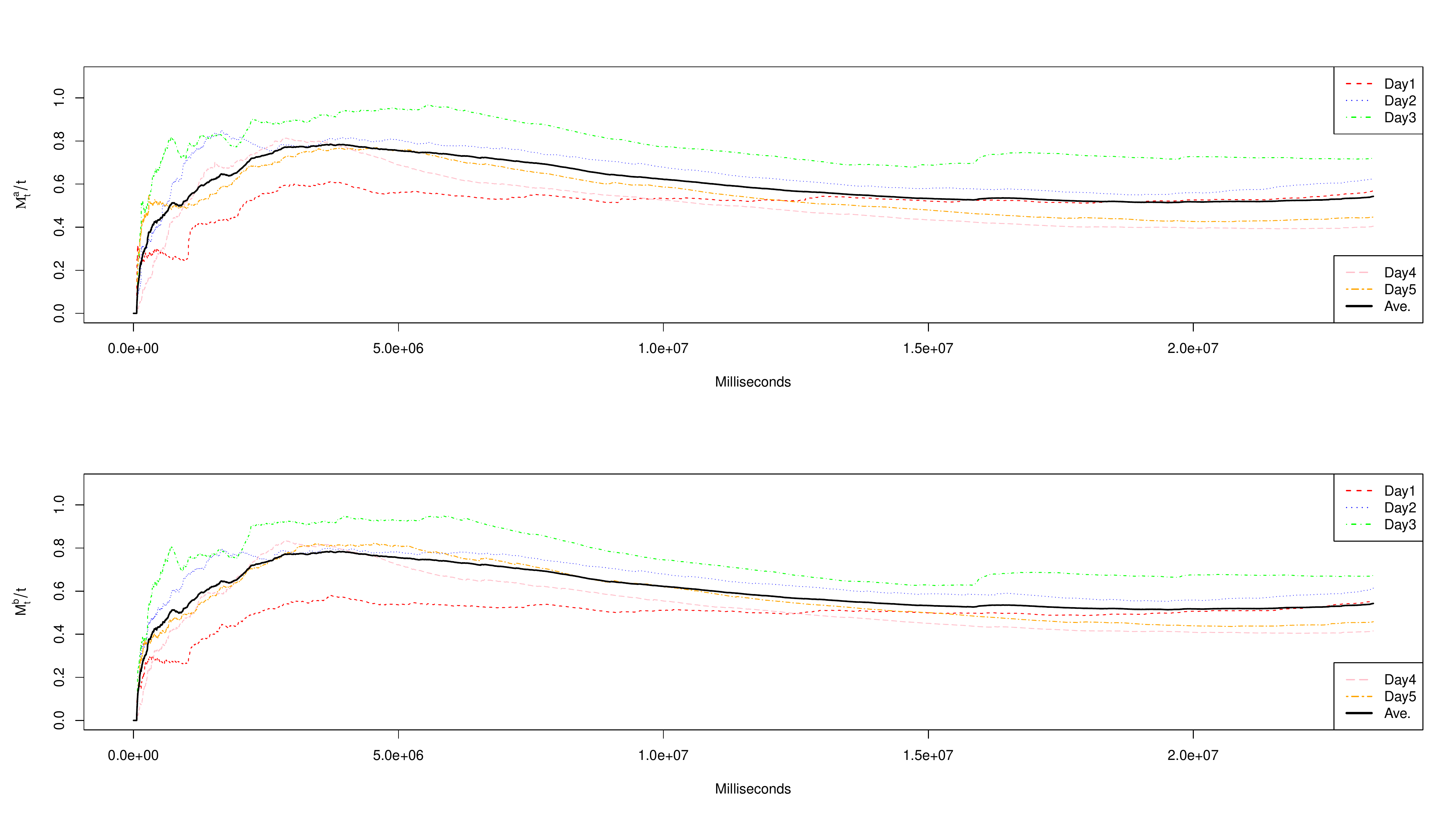}
\caption{Graphs of $M_{it}^a/t$ and $M_{it}^b/t$ for each of the five days.}\label{fig:Mu_over_T}
\end{center}
\end{figure}

\begin{figure}[h!]
\begin{center}
\includegraphics[height=3in,width=4in]{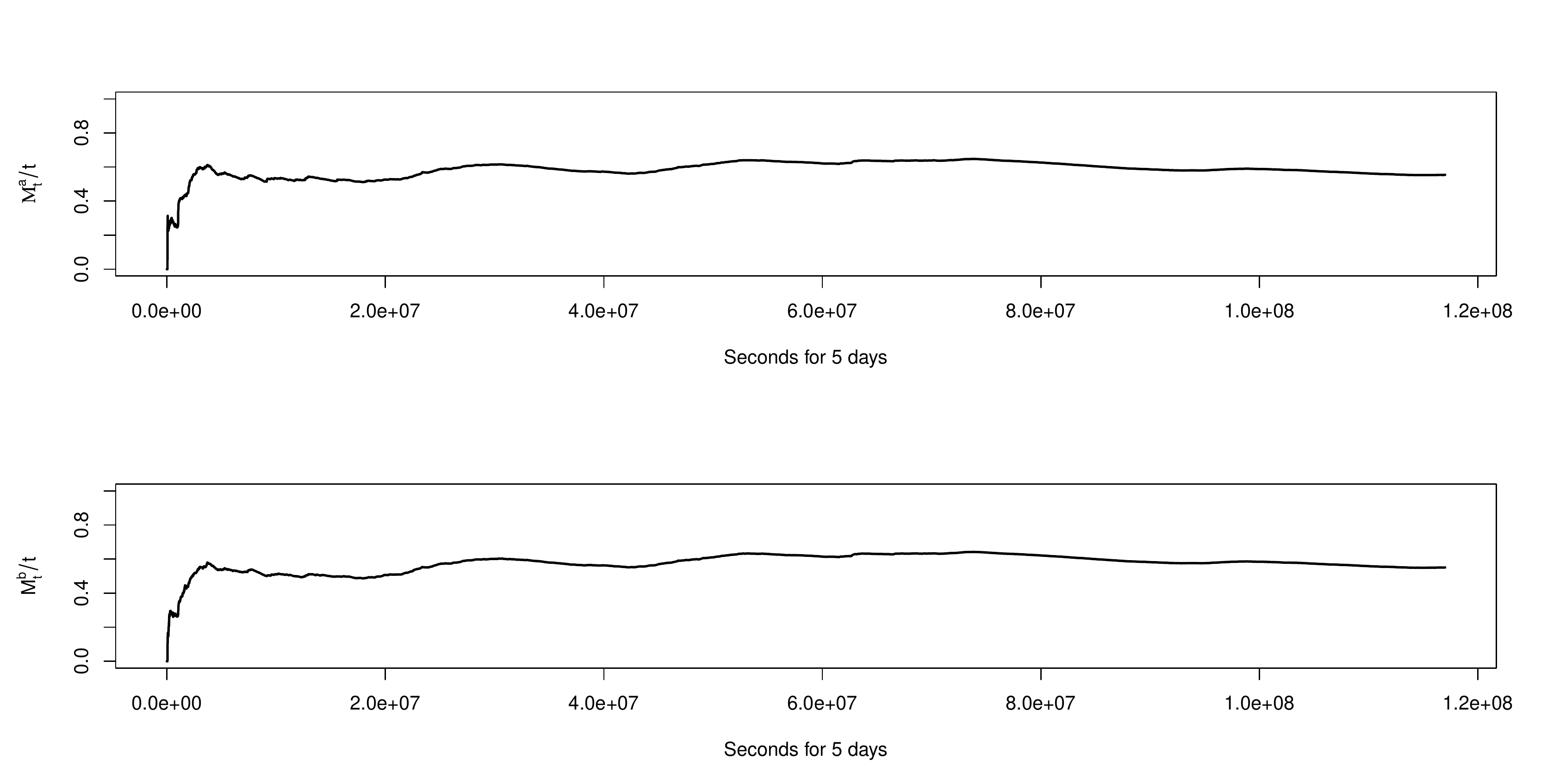}
\caption{Graphs of $M_{t}^a/t$ and $M_{t}^b/t$ for five days.}\label{fig:Mu_over_T_all}
\end{center}
\end{figure}

\newpage

\begin{figure}[h!]
\begin{center}
\includegraphics[height=2in,width=4in]{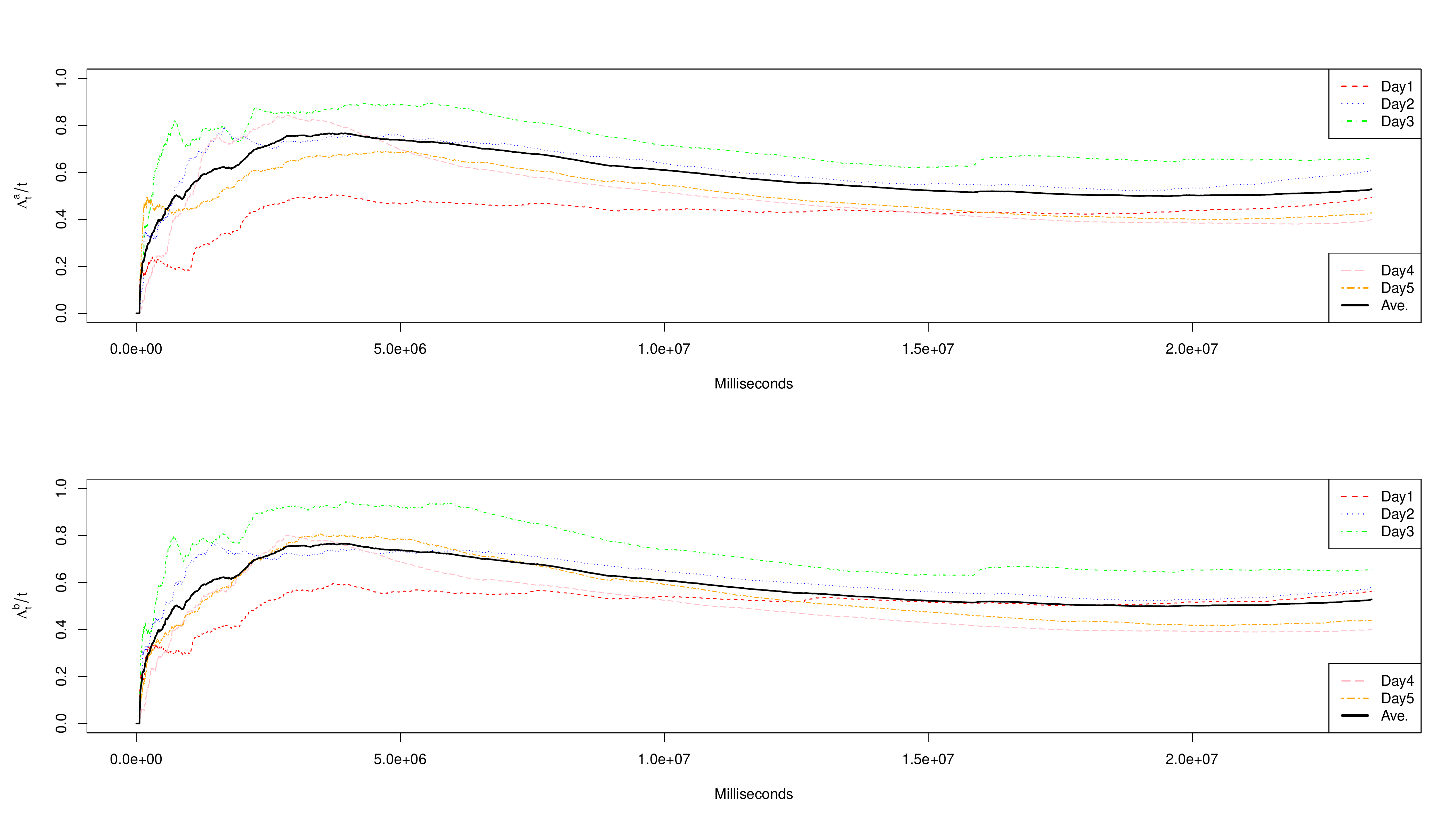}
\caption{Graphs of $\Lambda_{it}^a/t$ and $\Lambda_{it}^b/t$ for each of the five days.}\label{fig:Lambda_over_T}
\end{center}
\end{figure}

\begin{figure}[h!]
\begin{center}
\includegraphics[height=3in,width=4in]{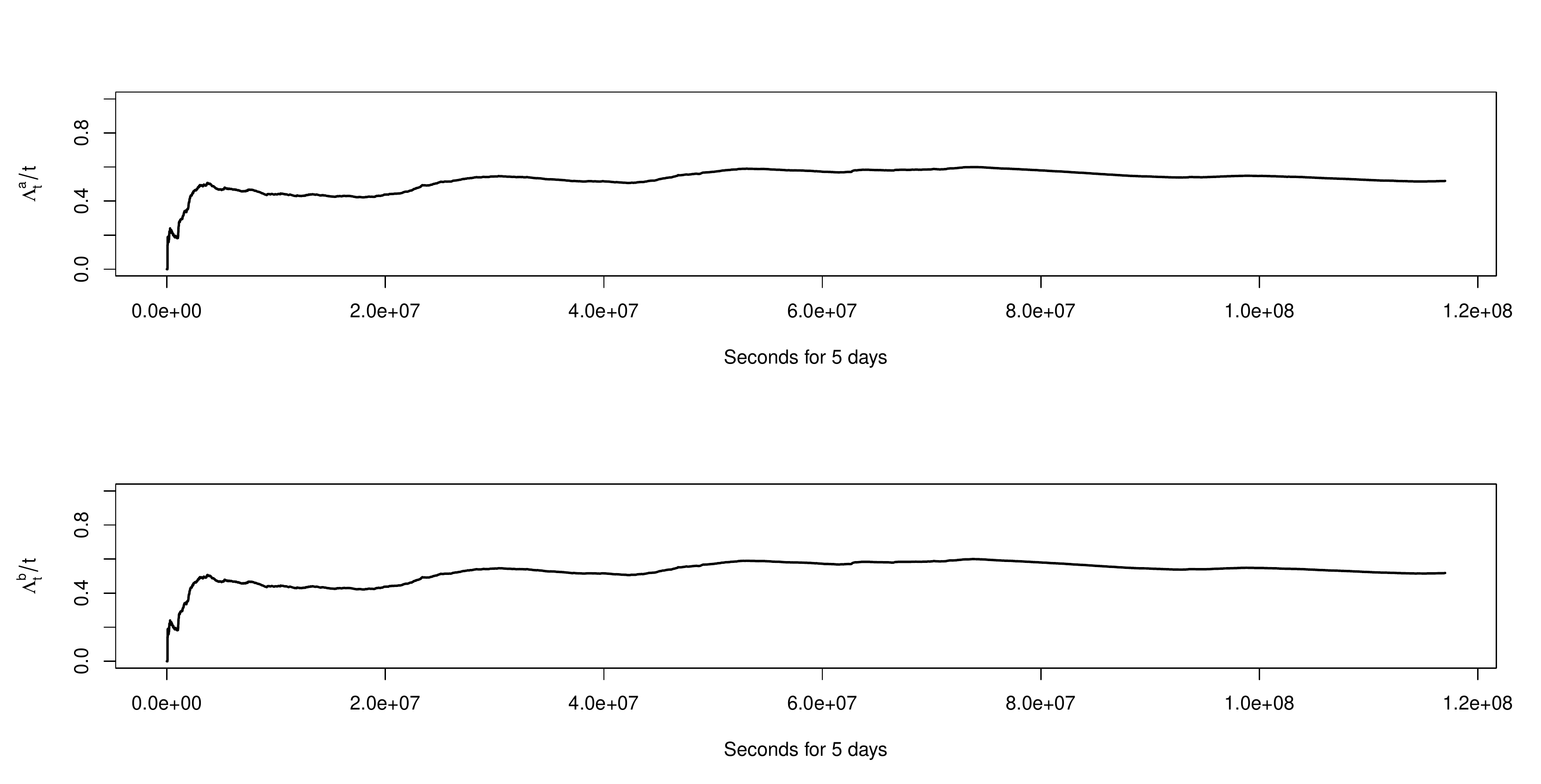}
\caption{Graphs of $\Lambda_{t}^a/t$ and $\Lambda_{t}^b/t$ for five days.}\label{fig:Lambda_over_T_all}
\end{center}
\end{figure}

\newpage

\begin{figure}[h!]
\begin{center}
\includegraphics[height=2in,width=4in]{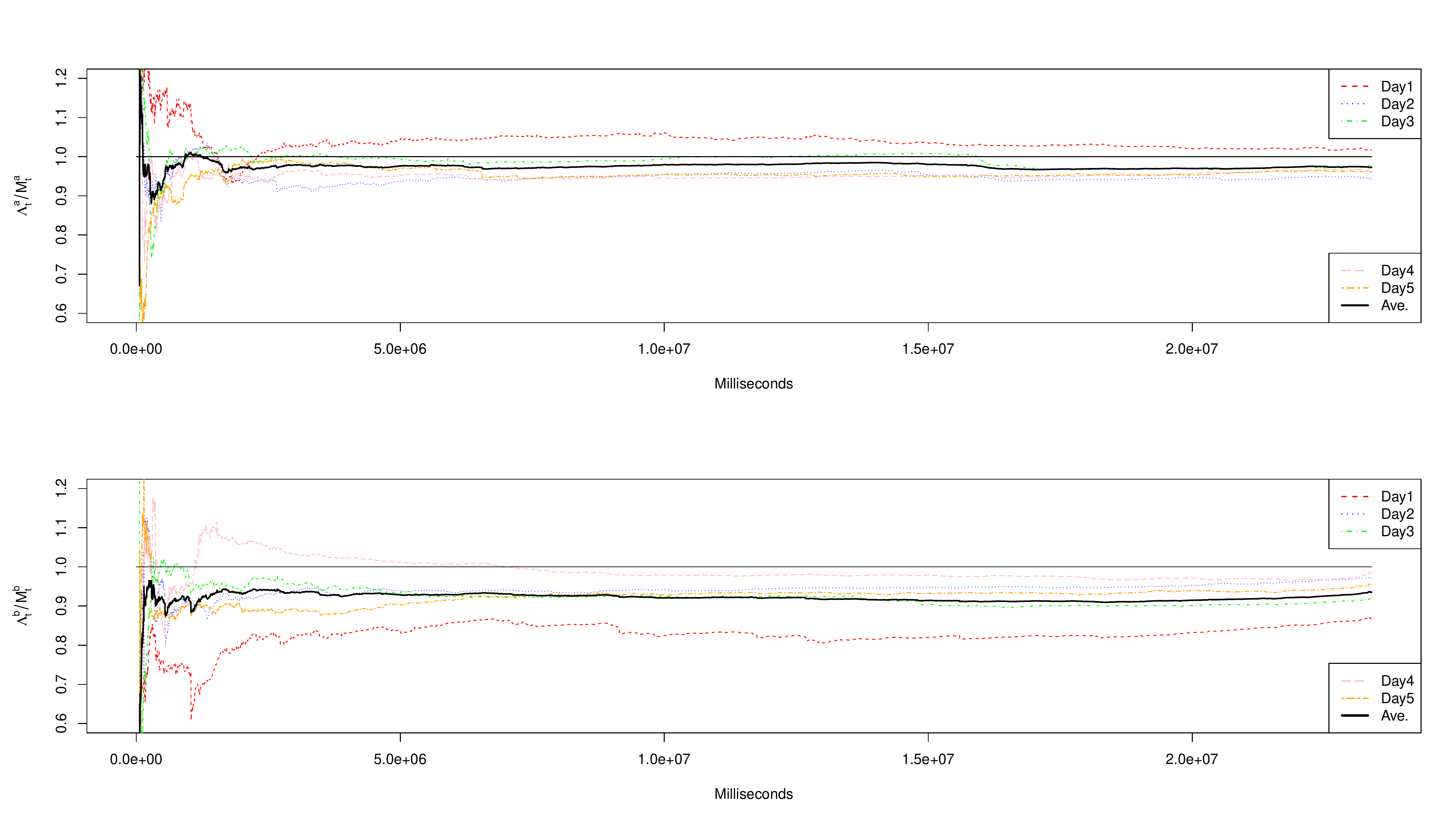}
\caption{Graphs of $\Lambda_{it}^a/M_{it}^a$  and $\Lambda_{it}^b/M_{it}^b$ for each of the five days.}\label{fig:Lambda_over_Mu}
\end{center}
\end{figure}

\begin{figure}[h!]
\begin{center}
\includegraphics[height=3in,width=4in]{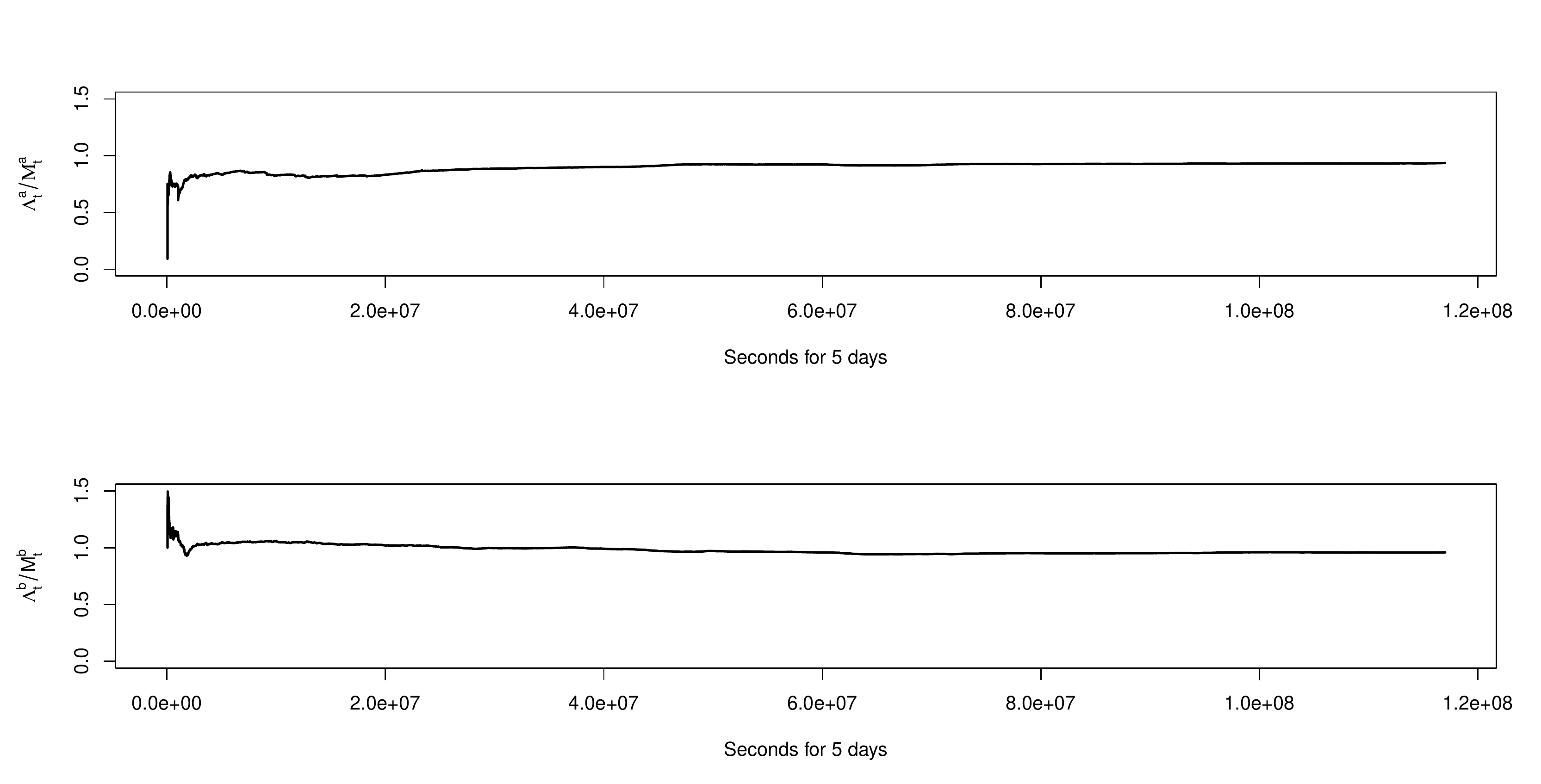}
\caption{Graphs of $\Lambda_{t}^a/M_{t}^a$  and $\Lambda_{t}^b/M_{t}^b$ for five days.}\label{fig:Lambda_over_Mu_all}
\end{center}
\end{figure}

\begin{figure}[h!]
\begin{center}
\includegraphics[height=2in,width=4in]{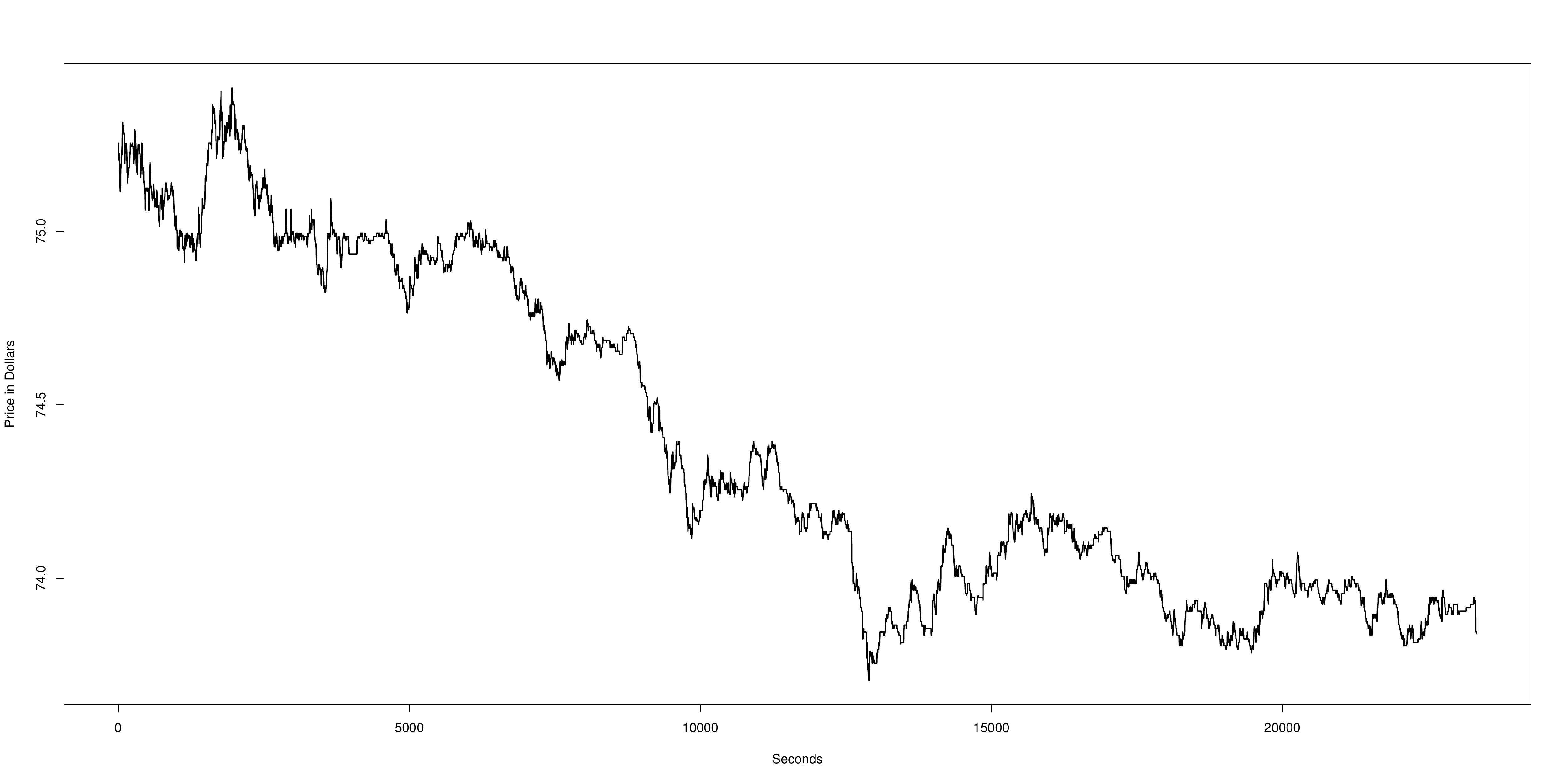}
\caption{Graphs of the midprice for November 3rd, 2014.}\label{fig:midprice}
\end{center}
\end{figure}

\newpage

\bibliographystyle{apalike}

\def\cprime{$'$}

\appendix

\section{Auxiliary results}\label{app:aux}
\begin{prop}\label{prop:laplace}
Suppose that $V_n = X_1+ \cdots +X_n$, where the variables $X_i$ are i.i.d. with $xP(X_i > x) \stackrel{x\to\infty}{\to} c \in (0,\infty)$. Then $\frac{V_n}{n\log{n}} \stackrel{Pr}{\to} c$, as $n\to\infty$.
\end{prop}

\begin{proof}
First, for any $s>0$ and $T>0$,
$$
s \int_T^\infty \frac{e^{-sx}}{x} dx = s \int_{sT}^\infty \frac{e^{-y}}{y} dy = - s\log(Ts) e^{-Ts} + s \int_{Ts}^\infty \log(y) e^{-y}dy,
$$
so as $s\to 0$,  $ s \int_T^\infty \frac{e^{-sx}}{x} dx \sim -s\log{s}$.
Next, for any non negative random variable $X$ and any $s\ge 0$,
$$
\Ex\left[e^{-s X} \right] = 1 - s\int_0^\infty P(X>x)e^{-sx} dx.
$$
As a result, if $P(X > x) \sim c/x$, as $x\to\infty$, 
then, as $s\to 0$,
$$
\Ex\left[e^{-s X} \right] = 1 +c s \log{s} + o(s\log{s}).
$$
Therefore, setting $a_n = n\log{n}$, one obtains, for a fixed $s>0$,
\begin{eqnarray*}
\Ex\left[e^{-s V_n/a_n} \right]  &=& \left[ \Ex\left[e^{-s X_1/a_n} \right] \right]^n \\
&=&  \left\{1 - \frac{sc}{a_n} \log(sa_n) + o\left(log(a_n)/a_n\right)      \right\}^n  \\
&\stackrel{n\to\infty}{\to}&  e^{-cs},
\end{eqnarray*}
since $\frac{ns}{a_n} \log(sa_n) \to s$ as $n\to\infty$. Hence,  $V_n/a_n \stackrel{Pr}{\to} c$, as $n\to\infty$.
\end{proof}

\begin{prop}\label{prop:counting}
Suppose that $V_n/f(n) \stackrel{Pr}{\to} c$, as $n\to\infty$, where $f(n) \to \infty$ is regularly varying of order $\alpha$. Define $N_t = \max\{n \ge 0; \; V_n \le t\}$ and suppose that for some function $g$ on $(0,\infty)$, $f\circ g(t) \sim g\circ f(t)\sim t$, as $t\to\infty$. Then  $N_t/g(t) \stackrel{Pr}{\to} c^{-1/\alpha}$.
\end{prop}

\begin{proof} The proof is similar to the proof of the renewal theorem in \citet{Durrett:1996}[Theorem 7.3].
By definition,  $V_{N_t} \le t < V_{N_t+1}$. As a result,
$$
\frac{V_{N_t}}{f(N_t)} \le \frac{t}{f(N_t)} < \frac{V_{N_t}}{f(N_t+1)} \frac{f(N_t+1)}{f(N_t)}.
$$
By hypothesis, $V_n/f(n)$ converges in probability to $c \in (0,\infty)$, as $n\to\infty$. Also, since $V_n$ is finite for any $n\in \dN$, it follows that $N_t$ converges in [probability to $+\infty$ as $t\to\infty$. Next, since $f(n+1)/f(n)\to 1$ as $n\to\infty$, it follows that as $t\to\infty$, $f(N_t)/t$ converges in probability to $\frac{1}{c}$. Also, $g$ is regularly varying of order $1/\alpha$, so one may conclude that $N_t/g(t) \stackrel{Pr}{\to} c^{-1/\alpha}$.
\end{proof}
\begin{rem}\label{rem:inverse}
If $f(t) = t\log{t}$, then $\alpha=1$ and one can take $g(t) = t/\log{t}$.
\end{rem}

\begin{prop} \label{prop:psi}
Set $\psi_{\lambda}(t,x) =\int_t^\infty \frac{1}{u}I_x(2u\lambda)e^{-2u\lambda}du$, for any $t,x,\lambda>0$. Then there exists a constant $C$ so that for any $x,\lambda>0$, and any $t\ge \frac{1}{2\lambda}$,
$
\psi_\lambda(t,x) \le \frac{C}{\sqrt{2\lambda t}}.
$
\end{prop}

\begin{proof}

First, note that $\psi_{\lambda}(t,x) = \psi_{1/2}(2\lambda t,x)$. It is well-known that
\begin{eqnarray*}
I_x(z) &=&\frac{1}{\pi}\int_0^\pi e^{z\cos\theta}\cos(x\theta)d\theta \le \frac{1}{\pi}\int_0^\pi e^{z\cos\theta}d\theta \\
& \le & \frac{1}{2}+ \frac{1}{\pi}\int_0^1 \frac{e^{zs}}{\sqrt{1-s^2}}ds.
\end{eqnarray*}
Next, set
$
E_1(u):=\int_{u}^{\infty}\frac{e^{-w}}{w} dw$, $ u\more0$. Then
\begin{eqnarray*}
\psi_{1/2}(t,x) &\le &  \int_t^\infty \frac{e^{-u}}{u} \left\{\frac{1}{2}+ \frac{1}{\pi}\int_0^1 \frac{e^{us}}{\sqrt{1-s^2}}ds\right\}du\\
&=& \frac{1}{2}E_1(t) +\frac{1}{\pi}\int_t^\infty\int_0^1 \frac{e^{-su}}{u\sqrt{s(2-s)}}ds du\\
&=&  \frac{1}{2}E_1(t) +\frac{1}{\pi}\int_0^1 \frac{E_1(st)}{\sqrt{s(2-s)}}ds.\\
&=&  \frac{1}{2}E_1(t) +\frac{1}{\pi}\int_0^t \frac{E_1(s)}{\sqrt{s(2t-s)}}ds.
\end{eqnarray*}
According to  \citet[Section 6.8.1]{Olver/Lozier/Boisvert/Clark:2010},  $E_{1}(u)\le e^{-u}\ln\left(1+1/u\right)$ for any $u>0$. Furthermore,
$\ln(1+x)\leq x$ and $\ln(1+x)\leq x^{2/5}$  for any $x\ge 0$.  As a result,
\begin{eqnarray*}
\psi_{1/2}(t,x) &\le &   \frac{e^{-t}}{2t}  +\frac{t^{-1/2}}{\pi}\int_0^t s^{-9/10} e^{-s}ds \le  \frac{e^{-t}}{2t}  +\frac{\Gamma(\frac{1}{10})}{\pi t^{1/2}} \le Ct^{-1/2}
\end{eqnarray*}
for any $t\ge 1$, where $C = \frac{e^{-1}}{2}+ \frac{\Gamma(\frac{1}{10})}{\pi}$.
\end{proof}

\section{Proofs}\label{app:proofs}

\begin{proof}[Proof of Lemma \ref{lemma:tail:sigma}] From \cite{Olver/Lozier/Boisvert/Clark:2010}[Formula 10.30.4], for fixed $\nu$,
$
I_\nu(z)\sim \frac{e^z}{\sqrt{2\pi z}}$ as $z\to\infty$. Also, from \citet[p. 376]{Abramowitz/Stegun:1972}, $ I_n(z)= \frac{1}{\pi}\int_0^\pi e^{z \cos{\theta}}\cos(n\theta)d\theta 	$, so
for any $x\in \dN$, $I_n(z)\le e^{z}$.
Thus, as $T\to\infty$,
\begin{eqnarray*}
\Px_x[\sigma_Y\more T]&= &\left(\frac{\mu}{\lambda}\right)^{x/2}\int_T^{\infty}\frac{x}{s}I_x\left(2s\sqrt{\lambda\mu}\right)e^{-s(\lambda+\mu)}ds\\
		&\sim &\left(\frac{\mu}{\lambda}\right)^{x/2}\int_T^{\infty}\frac{x}{s} \frac{e^{2s\sqrt{\lambda\mu}}}{\sqrt{4s\pi\sqrt{\lambda\mu}}}e^{-s(\lambda+\mu)}ds\\
		&\sim &\left(\frac{\mu}{\lambda}\right)^{x/2} \frac{x}{2\sqrt{\pi\sqrt{\lambda\mu}}} \int_T^{\infty}s^{-3/2}e^{-s\cC}ds.
\end{eqnarray*}
Also, for any $x\in \dN$,
\begin{equation}\label{eqn:upperbound:survival}
\Px_x[\sigma_Y\more T] \le x \left(\frac{\mu}{\lambda}\right)^{x/2} \int_T^{\infty}s^{-1}e^{-s\cC}ds.
\end{equation}
Consequently, if $\lambda=\mu$, $\cC=0$ and
\begin{align*}
\Px_x[\sigma_Y\more T] &\sim  \frac{x}{2\lambda\sqrt{\pi}} \int_T^{\infty}s^{-3/2}ds \sim \frac{x}{2\lambda\sqrt{\pi}} \frac{2}{\sqrt{T}} 		\sim \frac{x}{\lambda\sqrt{\pi T}} .
\end{align*}
This agrees with the result proved in \cite{Cont/deLarrard:2013}. However, if $\lambda\less\mu$, using the change of variable $u = s\cC$, one gets
\begin{align*}
\Px_x[\sigma_Y\more T] &\sim \cC^{1/2}\left(\frac{\mu}{\lambda}\right)^{x/2} \frac{x}{2\sqrt{\pi\sqrt{\lambda\mu} } } \int_{T\cC}^{\infty}{u^{-3/2}}e^{-u} du \\
		&\sim \left(\frac{\mu}{\lambda}\right)^{x/2} \frac{x}{\sqrt{\pi\sqrt{\lambda\mu}}} \left[\frac{e^{-T\cC}}{\sqrt{T}} - \sqrt{\cC} \Gamma\left(\frac{1}{2},T\cC\right)\right].
\end{align*}
To compute the expectation in the case where $\lambda=\mu$, note that for large enough $T$,
$
\Ex_x\left[\sigma_Y\right]= \int_0^\infty \Px_x[\sigma_Y\more t] dt \geq
\frac{x}{2\lambda\sqrt{\pi}} \int_T^\infty \frac{1}{\sqrt{t}}dt = \infty$,
whereas if $\lambda\less \mu$, for a sufficiently large $T$, there are finite constants $C_1$ and $C_2$ such that for any $0 \le \theta <\cC$,
\begin{align*}
\Ex_x\left[e^{\theta\sigma_Y} \right]&= 1+  \theta \int_0^\infty e^{\theta t} \Px_x[\sigma_Y\more t] dt \leq C_1 + \theta C_2\int_T^\infty  e^{-t(\cC-\theta)} dt\\
		&=C_1 + C_2 \frac{e^{-T(\cC-\theta)}}{(\cC-\theta)}  \less\infty.
\end{align*}
\end{proof}

\begin{proof}[Proof of Proposition \ref{prop:Distr:Sn}]
Let $F_{n,Q}(t;x,y)$ and $F_{n,L}(t;x,y)$ denote the cdf of $S^n_{Q}$ and $S^n_{L}$, respectively, starting from $z_0= (x,y)$,  with densities $f_{n,\cQ}(t;z_0)$ and $f_{n,\cL}(t;z_0)$, where
$F_{n,Q}(\cdot; z_0)$ is the convolution of $F_{1,Q}$ $(n-1)$ times with $F_{1,Q}(\cdot; z-0)$. The result will be proven by induction. The base case $n=1$ is given in Corollary \ref{cor:Distr:tau}. Assume the result is true for any $m\leq n\in\N$. Then by Corollary \ref{cor:Distr:tau} and the induction hypothesis,
\begin{equation}\label{eqn:rel:Cx:Ax}
F_{\cL}(t;x,y)=F_{\cQ}(A_t;x,y) \text{ and } f_{n,\cL}(t;x,y)=f_{n,\cQ}(A_t;x,y)\alpha_t.
\end{equation}
Also, by the definition of $\tau_n$ and $V_n$, under Assumption \ref{hyp:ftilde}, if $z_0=(x,y)$, then
\begin{align*}
F_{n,\cL}(t;z_0) &= \Px_{\cL}[V_{n+1} \leq t |\; q_0 = z_0 ]=\Px_{\cL}[V_{n}\leq t, \tau_{n+1}\leq t-V_{n} | \; q_0=z_0]\\
		&= \sum_z f(z) \int_0^{t} \Px_{\cL}[\tau_{n+1}\leq t-u| q_u=z]f_{n,\cL}(u;z_0)du\\
		&= \sum_z f(z)\int_0^{t} \Px_\cQ\left[\tau_{n+1}\leq A_{t-u}^{(n+1)}|q_u=z\right]f_{n,\cQ}(A_u;z_0)\alpha_u du\\
		&=\int_0^{t} F_{1,\cQ}(A_t-A_u)f_{n,\cQ}(A_u;z_0)\alpha_u du = \int_0^{A_t} F_{1,\cQ}(A_t-u)f_{n,\cQ}(u;z_0) du\\
		&=\int_0^{A_t} F_{1,\cQ}(A_t-u)dF_{n,\cQ}(u;z_0) =\int_0^{A_t} F_{n,\cQ}(A_t-u)dF_{1,\cQ}(u;z_0)\\
		&=\Px_{\cQ}\left[V_{n+1} \leq A_t| \; q_0=z_0\right],
\end{align*}
where we used the fact that for any $s\ge 0$,  $\alpha^{(n+1)}(s) = \alpha(s+u)$  given $V_{n}=u$, so $A^{(n+1)}(t) = \int_0^{t} \alpha(s+u)ds = A_{t+u}-A_u$.
Furthermore, in the last equality we used  the fact that for $X$ and $Y$, non-negative independent random variables,
\[
F_{X+Y}(t)=\Px[X+Y\leq t]=F_X*F_Y(t)=\int_0^tF_X(t-x)dF_Y(x),
\]
with $F_X$ and $F_Y$ denoting the cdfs of $X$ and $Y$. Furthermore, starting $q_0$ from distribution $f$, one obtains that
$\Px_\cL[V_n \le t] = \Px_\cQ[V_n \le A_t]$.
\end{proof}

\end{document}